\documentclass[a4paper,10pt]{llncs}

\usepackage[english]{babel}
\usepackage{hyperref}

\usepackage{latexsym,amsfonts,amsmath,amssymb,amscd,stmaryrd,mathrsfs}

\usepackage[dvips]{epsfig}
\usepackage{graphicx,graphics,psfrag}
\usepackage{float}
\usepackage{color}
\usepackage[table]{xcolor}

\definecolor{violet}{RGB}{208,32,144}

\newcommand{\define}{\emph}
\newcommand{\A}{\mathcal{A}}
\newcommand{\B}{\mathcal{B}}
\newcommand{\G}{\mathbb{G}}
\newcommand{\N}{\mathbb{N}}
\newcommand{\s}{\sigma}
\renewcommand{\S}{\mathscr{S}}
\newcommand{\Seq}{\mathrm{S}}
\newcommand{\supp}{\mathrm{supp}}
\newcommand{\T}{\mathbf{T}}
\newcommand{\U}{\mathbb{U}}
\newcommand{\Z}{\mathbb{Z}}

\newcommand{\vect}[1]{\mathbf{#1}}

\newcommand{\sa}[2]{\mathbf{SA}_{#1}\left(#2\right)}

\newcommand{\reshift}{\mathcal{RE}}

\def\SA{\mathbf{SA}}

\usepackage{wasysym}
\usepackage{marvosym,fancybox}

\newcounter{claimcount}[theorem]
\newcommand{\THMfont}[1]{{\sl #1}}
\newcommand{\Claim}[1]{\refstepcounter{claimcount} \vspace{0.3em}              
\noindent {\sc Claim \theclaimcount: \ }\THMfont{ #1}}
\newcommand{\bprf}[1][Proof:]{\begin{list}{}    {\setlength{\leftmargin}{0.5em}
\setlength{\rightmargin}{0em}  \setlength{\listparindent}{1em}}   \item {\em
\hspace{-1em}  #1  }}
\newcommand{\eprf}{\end{list}}

\newcommand{\bclaimprf}{\bprf}

\newcommand{\eclaimprf}{ \hfill $\Diamond$~{\scriptsize {\tt
Claim~\theclaimcount}}\eprf} % %

\usepackage{tikz}
\usetikzlibrary{patterns,positioning,arrows,decorations.pathmorphing}

\definecolor{rouge}{RGB}{255,77,77}
\definecolor{vert}{RGB}{0,178,102}
\definecolor{jaune}{RGB}{255,255,0}
\definecolor{violet}{RGB}{208,32,144}
\definecolor{orange}{RGB}{255,140,0}
\definecolor{bleu}{RGB}{0,0,205}

\title{Multidimensional effective S-adic systems\\ are sofic.}
\author{Nathalie Aubrun and Mathieu Sablik% \\
}
\institute{LIGM, Universit\'e Paris-Est and LATP, Universit\'e de Provence}
\date{}

\begin{document}
\maketitle

\begin{abstract}
In this article we prove that multidimensional effective S-adic systems,
obtained by applying an effective sequence of substitutions chosen among a
finite set of substitutions, are sofic subshifts. 
\end{abstract}

\section*{Introduction}

Let $\A$ be a finite alphabet. A $d$-dimensional subshift $\T\subset\A^{\Z^d}$
is a closed and shift-invariant set of configurations, where the shift is the
natural action of $\Z^d$ on the configurations space $\A^{\Z^d}$. With a
combinatorial point of view, one can equivalently define subshifts by excluding
configurations that contain some forbidden finite patterns. Depending on the
conditions imposed on this set of forbidden patterns, it is possible to define
several classes of subshifts. The simplest one is the class of subshifts of
finite type (also called SFT), where the set of forbidden finite patterns may be
chosen finite. A larger class is the one of sofic subshifts, which are images of
SFT under a factor map. These two classes are defined locally and they are well
understood in dimension 1.

A way to construct minimal aperiodic subshifts is to consider subshifts
generated by a fix point of substitution, introduced in dimension one
by Thue~\cite{Thue1906} and generalized to higher dimensions. These subshifts
constitute the class of the substitutive subshifts. More precisely, for a
substitution $s$ one can consider the subshift $\T_{\{s\}}$ where the allowed
patterns are given by iterations of the substitution $s$ on a letter of $\A$, or
$\T'_{\{s\}}$ the set of configurations which have pre-images by arbitrarily
many
iterations of $s$. Of course $\T_{\{s\}}\subset\T'_{\{s\}}$. In
dimension 1 the class of substitutive subshifts and the class of sofic subshifts
are disjoint: substitutive subshifts have low complexity~\cite{Pansiot84}, while
the only sofic subshifts with low complexity are periodic. In the
multidimensional framework the situation is different since all substitutive
subshifts are sofic. This result is a generalization to any
substitution~\cite{mozes1989tss} of the original construction of aperiodic
tilings~\cite{robinson1971uan}.

A possible generalization of the construction of substitutive subshifts is to
consider S-adic subshifts, that were introduced by S. Ferenczi in the
one-dimensional setting~\cite{Ferenczi96}. Given a finite set of substitutions
$\S$, and a sequence $\Seq\in\S^\N$, we define the subshifts $\T_\Seq$ and
$\T'_\Seq$ where the iterations of the different substitutions are given by the
sequence $\Seq$. This class of subshifts is studied in dimension 1 and, under
some condition on the set $\S$, it is shown that the complexity is
low~\cite{Ferenczi96,Durand00}. It is thus natural to wonder if there exists
sofic S-adic subshifts in higher dimensions. An argument of cardinality shows
that the class of S-adic multidimensional subshifts is not included in the class
of sofic subshifts. Indeed the class of sofic subshifts is countable while there
are uncountably many ways to chose an infinite sequence of $\S$. The purpose of
this article is to show that S-adic subshifts which are sofic are exactly those
for which the sequence $\Seq$ is effective. 

The main idea of the proof is to use the result of S. Mozes~\cite{mozes1989tss}
which proves that a substitutive subshift is sofic in the case where the
substitution is not deterministic. This means that each time one wants to use a
substitution, it is possible to choose a rule among a set of substitutions $\S$.
However, contrary to the S-adic subshifts, at each level of iteration different
substitutions of $\S$ may appear. The aim of the proof is to synchronize these
substitutions, and in that purpose we need to introduce a one dimensional
effective subshift which codes the sequence of substitutions. This effective
subshift can
be realized by a $3$-dimensional sofic subshift thanks to the result of M.
Hochman~\cite{hochman2007drp} or by a $2$-dimensional  sofic subshift thanks to
the improvement obtained by~\cite{DurandRomashchenkoShen09}
or~\cite{AubrunSablik10}.

\section{Definition and classical properties}\label{section.def}

\subsection{Notion of subshift}\label{subsection.subshift}

Let $\A$ be a finite alphabet and $d$ be a positive integer. A
\define{configuration} $x$ is an element of $\A^{\Z^d}$. Let $\U$ be a
finite subset of $\Z^d$, we denote by $x_{\U}$ the \define{restriction}
of $x$ to $\U$. A \define{$\Z^d$-dimensional pattern} is an element
$p\in\A^{\U}$ where $\U\subset\Z^d$ is finite, $\U$ is the \define{support} of
$p$, which is denoted by $\supp(p)$. A pattern $p$ of support $\U\subset\Z^d$
\define{appears} in a configuration $x$ if there exists $i\in\Z^d$ such that
$p_{\U}=x_{i+\U}$, and in this case we note $p\sqsubset x$.
% The \define{elementary support of size n}, denoted by $\U_n$, is the support
% $[-n;n]^d$ for any integer $n\in\N$. A pattern whose support is $\U_n$ is
% called
% an \define{elementary pattern of size $n$}.

We define a topology on $\A^{\Z^d}$ by endowing $\A$ with the discrete topology,
and considering the product topology on $\A^{\Z^d}$. For this topology,
$\A^{\Z^d}$ is a compact metric space on which $\Z^d$ acts by translation via
$\s$ defined for every $i\in\Z^d$ by:
$$
\s_{\A}^i:
\left(\begin{array}{cccc}
\A^{\Z^d} & \longrightarrow &\A^{\Z^d}&\\
x&\longmapsto & \s_{\A}^i(x)& \textrm{ such that } \s_{\A}^i(x)_u=x_{i+u} \
\forall u\in\Z^d
\end{array}\right).
$$

The $\Z^d$-action $(\A^{\Z^d},\s)$ is called the \define{fullshift}. Let
$\T\subset\A^{\Z^d}$ be a closed subset $\s$-invariant, the $\Z^d$-action
$(\T,\s)$ is a \define{subshift}. 

Let $F$ be a set of finite patterns, we define the \define{subshift of forbidden
patterns} $F$ by
$$\T_F=\left\{x\in\A^{Z^d}: \forall p\in F,\, p\textrm{ does not appear in }x
\right\}.$$
 
It is well known that every subshift can be defined by this
way~\cite{lind1995isd}. Let $\T$ be a subshift. If there exists a finite set of
forbidden patterns such that $\T=\T_F$, then $\T$ is a \define{ subshift of
finite type}. If there exists a recursively enumerable set of forbidden patterns
such that $\T=\T_F$, then $\T$ is an \define{effective subshift}.

\subsection{Factor and projective subaction}\label{subsection.operations}

Let  $(\A^{\Z^d},\s)$ and  $(\B^{\Z^d},\s)$ be two fullshifts. A \define{factor}
is a continuous function $\pi:\A^{\Z^d}\to\B^{\Z^d}$ such that
$\pi\circ\s=\s\circ\pi$. Let $\T\subset\A^{\Z^d}$ be an SFT and let
$\pi:\A^{\Z^d}\to\B^{\Z^d}$ be a factor, then $\pi(\T)\subset\A^{\Z}$ is a
subshift called a \define{sofic subshift}. In dimension $1$, sofic subshifts are
well understood since they are exactly subshifts whose language is
regular~\cite{lind1995isd}.

Let $\G$ be a sub-group of $\Z^d$ finitely and freely generated by
$u_1,u_2,\dots,u_{d'}$ ($d'\leq d)$. Let $\T\subseteq\A^{\Z^d}$ be a subshift,
the \define{projective subdynamic} -- or \define{projective subaction} -- of
$\T$ according to $\G$ is the subshift of dimension $d'$ defined by
$$\sa{\G}{\T}=\left\{y\in\A^{\Z^{d'}}\ :\  \exists x\in\T \textrm{ such that }
\forall i_1,\dots,i_{d'}\in\Z^{d'},
y_{i_1,\dots,i_{d'}}=x_{i_1u_1+\dots+i_{d'}u_{d'}}\right\}.$$

In~\cite{pavlovschraudner2009} the authors show that any $1$-dimensional sofic
subshift with positive entropy can be obtained as the projective subaction of a
$2$-dimensional SFT, but it remains open to know whether these subshifts are the
only ones that can be obtained in that way. A complete characterization was
obtained by Hochman~\cite{hochman2007drp} if we allow factor maps in addition to
projective subactions: the class of subshifts obtained by factor maps and
projective subactions of SFT is exactly the class of effective subshifts. The
original proof contains a construction that realizes any $1$-dimensional
effective subshift inside a $3$-dimensional SFT. This construction was improved
simultaneously by two different
techniques~\cite{AubrunSablik10,DurandRomashchenkoShen2010} to get
any $1$-dimensional effective subshift inside a $2$-dimensional SFT.

\begin{theorem}[\cite{hochman2007drp,AubrunSablik10,DurandRomashchenkoShen2010}]
\label{theorem.simulation}
Any effective subshift of dimension $d$ can be obtained with factor and
projective subaction operations from a subshift of finite type of strictly
higher dimension.
\end{theorem}

\section{Substitutive and S-adic subshifts}\label{section.sadic}

%%%
In this section we present substitutives and S-adic subshifts and prove that
multidimensional S-adic subshifts given by an effective sequence of
substitutions are sofic.

\subsection{Substitutions}

Let $\vect{n}=(n_1,\dots,n_d)\in\N^d$ and $\vect{k}=(k_1,\dots,k_d)\in\N^d$, we
define $\vect{n}+\vect{k}=(n_1+k_1,\dots,n_d+k_d)\in\N^d$,
$\vect{n}\otimes\vect{k}=(n_1.k_1,\dots,n_d.k_d)\in\N^d$ and
$\vect{n}^i=\vect{n}\otimes\dots\otimes\vect{n}$ with $i$ factors. Given
$\vect{k}=(k_1,\dots,k_d)$, we denote by $\U_{\vect{k}}$ the rectangle
$[0;k_1]\times[0;k_2]\times\dots\times[0;k_l]$. We say that $\vect{i}$ is
smaller (resp. strictly smaller) than $\vect{j}$ if for every $1\leq l \leq d$,
one has $i_l\leq j_l$ (resp. $i_l<j_l$). We denote it
by $\vect{i}\leq\vect{j}$ (resp. $\vect{i}<\vect{j}$).

Let $\A$ be a finite alphabet, we define the \define{set of rectangular pattern}
$\mathcal{P}=\bigcup_{\vect{k}\in\N^d}\A^{\U_\vect{k}}$. A
\define{$(\A,d)$-multidimensional substitution} of size $\vect{k}$ is a function
$s:\A\rightarrow \mathcal{P}$, for all $a\in\A$, we associate the vector
$\vect{k}^s(a)=(k^s_1(a),\dots,k^s_d(a))$ such that
$\supp(s(a))=\U_{\vect{k^s}(a)}$, that is to say the support of $s(a)$ depends
on the letter $a$. A $(\A,d)$-multidimensional substitution is non degenerate if
$k^s_l(a)\geq 1$ for every $l\in[1,d]$ and every $a\in\A$.

Let $(\vect{k}^n)_{n\in\Z}$ be a sequence of $d$-dimensional vectors. For each
$n\in\Z$ we define the function
$$
\phi^{(\vect{k}^n)_{n\in\Z}} :
\left(\begin{array}{ccl}
\Z^d & \rightarrow & \Z^d\\
\vect{i} & \mapsto &
\left(\phi_1(\vect{i}),\phi_2(\vect{i}),\dots,\phi_d(\vect{i})\right)           
\end{array}\right)
$$
where $\phi_i(r)=\sum_{j=0}^{r} (\vect{k}^j)_i$ if $r\geq 0$ and
$\phi_i(r)=\sum_{j=-1}^{r} (\vect{k}^j)_i$ if $r<0$.

Let $p\in\A^{\U_\vect{k}}$ be a rectangular pattern with finite support
$\U_\vect{k}\subset\Z^d$. We say that the substitution $s$ is
\define{compatible} with the pattern $p$ (resp. the configuration $x$) if for
all $\vect{i}=(i_1,\dots,i_d)\in\U$ and $\vect{j}=(j_1,\dots,j_d)\in\U$  (resp.
$\vect{i}=(i_1,\dots,i_d)\in\Z^d$ and $\vect{j}=(j_1,\dots,j_d)\in\Z^d$) such
that $i_l=j_l$ for one $l\in[1,d]$, one has
$k^{\textbf{s}}_l(p_\vect{i})=k^{\textbf{s}}_l(p_\vect{j})$.  Given a
substitution $s$ compatible with a configuration $x\in\Z^d$, one can transform
$\Z^d$ into a non-regular grid thanks to the function
$\phi^x=\phi^{(\vect{k}^s(x_{(n,\dots,n)}))_{n\in\Z}}$ (see
Figure\ref{figure.expand_grid}).

\begin{figure}[H]
\centering
\begin{tikzpicture}[scale=0.35]

% s^j(x)
\draw[very thick,color=black!10,fill=black!10,decorate,decoration={random steps,segment length=6pt,amplitude=3pt}](0,0) rectangle (4,4);
\draw (-2,2) node{$\s^{(3,-2)}(x)=$};
\draw [dashed] (2,0) -- (2,4);
\draw [dashed] (0,2) -- (4,2);
\draw (0.5,3) node{$\blacktriangle$};
\draw (2,2) node{$\bullet$};

%s^\sum(s(x))
\draw[very thick,color=black!10,fill=black!10,decorate,decoration={random steps,segment length=6pt,amplitude=3pt}](12,-2) rectangle (20,6);
\draw (8,2) node{$\s^{\phi^x\left((3,-2)\right)}[s_\infty(x)]=$};
\draw [dashed] (16,-2) -- (16,6);
\draw [dashed] (12,2) -- (20,2);

\draw[very thick, fill=white] (12.5,3.5) rectangle (13.5,5);
\draw (11.5,4.25) node{\tiny{$s(\blacktriangle)$}};
\draw[fill=white] (13.5,3.5) rectangle (14.5,5);
\draw[fill=white] (14.5,3.5) rectangle (16,5);
\draw[fill=white] (16,3.5) rectangle (17,5);

\draw[fill=white] (12.5,2.5) rectangle (13.5,3.5);
\draw[fill=white] (13.5,2.5) rectangle (14.5,3.5);
\draw[fill=white] (14.5,2.5) rectangle (16,3.5);
\draw[fill=white] (16,2.5) rectangle (17,3.5);

\draw[fill=white] (12.5,2) rectangle (13.5,2.5);
\draw[fill=white] (13.5,2) rectangle (14.5,2.5);
\draw[fill=white] (14.5,2) rectangle (16,2.5);
\draw[very thick, fill=white] (16,2) rectangle (17,2.5);
\draw (18,2.25) node{\tiny{$s(\bullet)$}};

% x
\draw[very thick,color=black!10,fill=black!10,decorate,decoration={random steps,segment length=6pt,amplitude=3pt}](0,9) rectangle (4,13);
\draw (-2,11) node{$x=$};
\draw [dashed] (2,9) -- (2,13);
\draw [dashed] (0,11) -- (4,11);
\draw (2,11) node{$\blacktriangle$};
\draw (3.5,10) node{$\bullet$};

% s(x)
\draw[very thick,color=black!10,fill=black!10,decorate,decoration={random steps,segment length=6pt,amplitude=3pt}](10,7) rectangle (19,15);
\draw (8,11) node{$s_\infty(x)=$};
\draw [dashed] (14,7) -- (14,15);
\draw [dashed] (10,11) -- (19,11);

\draw[very thick, fill=white] (14,11) rectangle (15,12.5);
\draw (13,11.75) node{\tiny{$s(\blacktriangle)$}};
\draw[fill=white] (15,11) rectangle (16,12.5);
\draw[fill=white] (16,11) rectangle (17.5,12.5);
\draw[fill=white] (17.5,11) rectangle (18.5,12.5);

\draw[fill=white] (14,10) rectangle (15,11);
\draw[fill=white] (15,10) rectangle (16,11);
\draw[fill=white] (16,10) rectangle (17.5,11);
\draw[fill=white] (17.5,10) rectangle (18.5,11);

\draw[fill=white] (14,9.5) rectangle (15,10);
\draw[fill=white] (15,9.5) rectangle (16,10);
\draw[fill=white] (16,9.5) rectangle (17.5,10);
\draw[very thick, fill=white] (17.5,9.5) rectangle (18.5,10);
\draw (19.5,9.75) node{\tiny{$s(\bullet)$}};

\end{tikzpicture}
\caption{If the configuration $x$ is compatible with the substitution $s$, then
one can define $\phi^x=\phi^{(\vect{k}^s(x_{(n,\dots,n)}))_{n\in\Z}}$.}
\label{figure.expand_grid}
\end{figure}
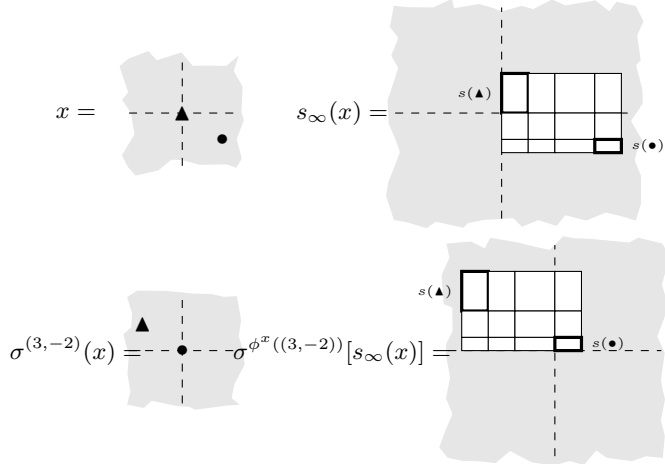

If the substitution $s$ is compatible with the configuration $x$ that contains a
pattern $p$, the substitution $s$ acts on $p$ and we obtain a pattern $s(p)$
whose support is 
$$\supp(s(p))=\bigcup_{\vect{i}\in\supp(p)}
\U_{\vect{k}^s(p_\vect{i})}+\phi^x(\vect{i})$$
and such that
$$\forall \vect{i}\in\supp(p),\forall \vect{j}\in\supp(s(p_\vect{i})),
s(p)_{\phi^x(\vect{i})+\vect{j}}=s(p_\vect{i})_\vect{j}.$$
So the substitution $s$ can easily be extended to a function on configurations
$s_\infty:\left( \begin{array}{ccc} \A^{\Z^d}&\to&\A^{\Z^d}\\ x&\mapsto&s(x)
\end{array}\right)$ where the if the substitution $s$ is compatible with the
configuration $x\in\A^{\Z^2}$,  the configuration $s_\infty(x)$ is defined as
explained above.

\begin{example}\label{example.1}
Let $\A$ be the two elements alphabet $\A=\{\circ,\bullet\}$ and $s$ be the
two-dimensional substitution whose rules are
$$
\footnotesize
\begin{array}{ccccccc}
\circ & \mapsto & \begin{array}{cc}\circ&\circ\\ \circ&\circ\end{array} &
\textrm{ and } & (\bullet,,) & \mapsto & \begin{array}{cc}\circ&\circ\\
\bullet&\circ\end{array}
\end{array}.
$$
For instance, the substitution $s$ acts on the pattern $p$ described below

\footnotesize
$$
\begin{array}{cccc}
s:
&
p=\begin{array}{ccc}
\circ&\bullet&\bullet\\
\bullet&\circ&\circ
\end{array}
&
\mapsto
&
s(p)=\begin{array}{cc|cc|cc}
\circ&\circ&\circ&\circ&\circ&\circ\\
\circ&\circ&\bullet&\circ&\bullet&\circ\\
\hline
\circ&\circ&\circ&\circ&\circ&\circ\\
\bullet&\circ&\circ&\circ&\circ&\circ
\end{array}
\end{array}.
$$
\normalsize

\end{example}

%%%

\subsection{Composition of substitutions}

Consider now that one wants to apply not only one but a finite set of
substitutions on a (finite or not) pattern $p$. We first define how to compute
the composition of two or more substitutions. Let $s,s'$ be two substitutions.
We say that $s'$ is compatible with $s$ if for any pattern $p$ compatible with
$s$, $s(p)$ is compatible with $s'$. If $s'$ is compatible with $s$, we can thus
define the composition $s'\circ s$. For a sequence of substitutions
$S_{[0,n]}=(s_0,\dots,s_n)$, one defines by induction the substitution
$\widehat{S_{[0,n]}}$ if $s_0$ is compatible with $\widehat{S_{[1,n]}}$ by
$\widehat{S_{[0,n]}}=s_0\circ\widehat{S_{[1,n]}} $. 

Let $\S$ be a finite set of $(\A,d)$-multidimensional substitutions. We present
the two classical points of view to make $\S$ act on the set of configurations
$\A^{\Z^d}$. In the first one, the set $\S$ acts on a configuration $x$ via a
sequence of substitutions $\Seq=(s_i)_{i\in\N}\in\S^\N$, and at iteration $i$
the substitution $s_i$ is applied to every letter in $x$ (see
Section~\ref{subsection.sadic}). In the second one, the set $\S$ acts on a
configuration $x$ in a non uniform way, that is to say at each iteration the
substitution applied depends on the position in $x$ (see
Section~\ref{subsection.nondeterm}).

%%%

\subsection{S-adic subshifts}\label{subsection.sadic}

Let $\S$ be a finite set of $(\A,d)$-multidimensional substitutions and let
$\Seq\in\S^{\N}$ be a sequence of substitutions. We want to define how this
sequence acts on a letter $a\in\A$. The principle is that at the iteration
number $i$, the substitution $s_0$ is applied to the whole pattern
$s_1\circ\dots\circ s_i(a)$. We define the two $S$-adic subshifts based on
this action of $\Seq$ on letters of $\A$
\begin{eqnarray*}
\T_{\Seq} & = &\left\{x\in\A^{\Z^d}: \forall p\sqsubset x, \, \exists
a\in\A,\exists n\in\N,\quad p\sqsubset \widehat{\Seq_{[0,n]}}(a) \right\}\\
\T'_{\Seq} & = &\left\{x\in\A^{\Z^d}: \forall n\in\N,\,\exists
y\in\A^{\Z^d},\quad \widehat{\Seq_{[0,n]}}(y)=x\right\}.\\
\end{eqnarray*}

The first subshift $\T_{\Seq}$ will be called the \define{local $S$-adic
subshift}. With the sequence of substitutions $\Seq$ we can produce patterns,
that are the $\widehat{S_{[0,n]}}(a)$ for any letter $a\in\A$, and these
patterns are seen are the allowed pattern of the subshift $\T_{\Seq}$. The
second subshift $\T'_{\Seq}$ will be called the \define{global $S$-adic
subshift}, and this time the idea is to consider only configurations
$x\in\A^{\Z^2}$ for which it is possible to find a pre-image of any order under
the sequence of substitutions $\Seq$. Obviously $\T_{\Seq}\subset\T'_{\Seq}$,
but the reciprocal inclusion does not necessary hold as Example~\ref{example.2}
shows.

\begin{example}\label{example.2}
Let $s$ be the substitution of Example~\ref{example.1}. Then if we choose
$\S=\{s\}$ and so $\Seq=s^\N$, the two $S$-adic subshifts defined above are in
this case 
$$
\begin{array}{ccccccc}
\T_{\Seq} & = &\left\{ \circ^{\Z^2} \right\} & \textrm{ and }& \T'_{\Seq} & =
&\left\{ \circ^{\Z^2}\right\}\cup\left\{\s^i(x_\bullet),i\in\Z^2 \right\}
\end{array}
$$

where the configuration $x_{\bullet}$ such that $x_{(i,j)}=\circ$ if
and only if $i\neq0$ and $j\neq0$ is in the subshift $\T'_{\Seq}$ -- $x_\bullet$
is a fixed-point of $s$ -- but not in the subshift $\T_{\Seq}$, since the
central pattern $x_{\U_1}$ appears neither in a
$\widehat{S_{[0,n]}}(\bullet)$ nor in a $\widehat{S_{[0,n]}}(\circ)$.
\end{example}

%%%

\subsection{Non-deterministic substitutions}\label{subsection.nondeterm}

Let $\S$ be a set of substitutions, another way to consider the action of this
set on a configuration in a non-deterministic (that is to say non-uniform) way.
That is to say, if we consider a pattern $\A^{\U}$, each letter can be modified
by a distinct substitution.

For a finite set $\U\subset\Z^d$, we consider the pattern of substitutions
$\textbf{s}\in\S^{\U}$. We say that the pattern of substitution
$\textbf{s}\in\S^{\U}$ is \define{compatible} with a pattern $p\in\A^{\U}$ if
for all $i=(i_1,\dots,i_d)\in\U$ and $j=(j_1,\dots,j_d)\in\U$ such that
$i_l=j_l$ for one $l\in[1,d]$, one has
$k^{\textbf{s}_i}_l(p_i)=k^{\textbf{s}_j}_l(p_j)$. 

If the pattern of substitution $\textbf{s}\in\S^{\U_\vect{k}}$ is compatible
with a pattern $p\in\A^{\U_\vect{k}}$ that appears in a configuration $x$, it
acts on $p$ and we obtain the pattern $\textbf{s}(p)$
\begin{itemize} 
\item whose support is $\displaystyle\supp(\vect{s}(p))=\bigcup_{i\in\supp(p)}
\U_{\vect{k}^{\vect{s}_i}(p_\vect{i})}+\phi^x(\vect{i})$
\item and such that $\forall \vect{i}\in\supp(p),\forall
\vect{j}\in\supp(\vect{s}_i(p_\vect{i})),
\vect{s}(p)_{\phi^x(\vect{i})+\vect{j}}=\vect{s}_i(p_\vect{i})_\vect{j}.$
\end{itemize}

\begin{example}\label{example.3}
Let $\S=\{s_a,s_b,s_c,s_d\}$ be a set of two-dimensional substitutions on the
alphabet $\A=\{ \circ,\bullet \}$ defined by the following rules
$$
\footnotesize
\begin{array}{cccccccc}
s_a : &  \circ & \mapsto & \begin{array}{cc}\circ&\circ\\ \circ&\circ\end{array}
& \textrm{ and } & \bullet & \mapsto & \begin{array}{cc}\circ&\circ\\
\bullet&\circ\end{array},
\end{array}
\begin{array}{cccccccc}
s_b : &  \circ & \mapsto & \begin{array}{ccc}\circ&\bullet&\circ\\
\circ&\bullet&\circ\end{array} & \textrm{ and } & \bullet & \mapsto &
\begin{array}{ccc}\circ&\circ&\circ\\ \bullet&\circ&\circ\end{array}
\end{array}
$$
$$
\footnotesize
\begin{array}{cccccccc}
s_c : &  \circ & \mapsto & \begin{array}{cc}\circ&\circ\\ \bullet&\circ\\
\circ&\bullet\end{array} & \textrm{ and } & \bullet & \mapsto &
\begin{array}{cc}\circ&\circ\\ \circ&\bullet\\ \bullet&\bullet\end{array}
\end{array},
\begin{array}{cccccccc}
s_d : &  \circ & \mapsto & \begin{array}{ccc}\bullet&\bullet&\bullet\\
\bullet&\bullet&\bullet\\ \circ&\circ&\circ\end{array} & \textrm{ and } &
\bullet & \mapsto & \begin{array}{ccc}\circ&\circ&\circ\\ \circ&\circ&\circ\\
\bullet&\bullet&\bullet\end{array}
\end{array}.
$$
Then given the pattern $p$ pictured below, the patterns of substitution
$\vect{s}$ and $\vect{s'}$ are compatible with $p$ and we can define the
patterns $\vect{s}(p)$ and $\vect{s'}(p)$, while the pattern of substitution
$\vect{s''}$ is not.  
$$
\footnotesize
p=\begin{array}{cccc}
\circ & \bullet & \bullet & \bullet \\
\bullet & \bullet & \circ & \circ
\end{array},~
\vect{s}=\begin{array}{cccc}
s_a & s_a & s_b & s_a \\
s_c & s_c & s_d & s_c
\end{array},~
\vect{s'}=\begin{array}{cccc}
s_a & s_a & s_a & s_a \\
s_c & s_c & s_c & s_c
\end{array},~
\vect{s''}=\begin{array}{cccc}
s_a & s_a & s_b & s_a \\
s_c & s_c & s_d & s_a
\end{array}
$$
$$
\footnotesize
\vect{s}(p)=
\begin{array}{cc|cc|ccc|cc}
\circ & \circ & \circ & \circ & \circ & \circ & \circ & \circ & \circ \\
\circ & \circ & \circ & \circ & \bullet & \circ & \circ & \bullet & \circ \\
\hline
\circ & \circ & \circ & \circ & \bullet & \bullet & \bullet & \circ & \circ \\
\circ & \bullet & \circ & \bullet & \bullet & \bullet & \bullet & \bullet &
\circ \\
\bullet & \bullet & \bullet & \bullet & \circ & \circ & \circ & \circ & \bullet 
\end{array},~
\vect{s'}(p)=
\begin{array}{cc|ccc|ccc|cc}
\circ & \circ & \circ & \circ & \circ & \circ & \circ & \circ & \circ & \circ \\
\bullet & \circ & \circ & \circ & \circ & \circ & \circ & \circ & \circ &
\bullet \\
\circ & \bullet & \bullet & \bullet & \bullet & \bullet & \bullet & \bullet &
\bullet & \bullet \\
\hline
\circ & \circ & \circ & \circ & \circ & \circ & \bullet & \circ & \circ & \circ
\\
\bullet & \circ & \bullet & \circ & \circ & \circ & \bullet & \circ & \circ
&
\circ 
\end{array}
$$
\end{example}

We define the set of \define{$\S$-patterns} by induction. A $\S$-pattern of
level
$0$ is an element of $\A$, and $p$ is a $\S$-pattern of level $n+1$ if there
exists a $\S$-pattern $p'\in\A^{\U}$ of level $n$ and a pattern of substitution
$\textbf{s}\in\S^{\U}$ compatible with $p'$ such that $\textbf{s}(p')=p$. Of
course the support of an $\S$-pattern is rectangular. The $\S$-patterns lead us
to define $\T_{\S}$, the \define{local subshift generated by the set of
substitutions~$\S$} $$\T_{\S}=\left\{x\in\A^{\Z^d}: \forall p\sqsubset x,\; p
\textrm{ is a
sub-pattern of a $\S$-pattern}\right\}.$$

Suppose that $\vect{s}\in\S^{\Z^d}$ is an infinite pattern of
substitutions and $x\in\A^{\Z^d}$ is a configuration. We denote by $\vect{s}(x)$
the configuration in $\A^{\Z^d}$ obtained if one applies $s_i$ on $x_i$ for
every $i\in\Z^d$. We thus define $\T'_{\S}$ the \define{global subshift
generated by the set of substitutions~$\S$} as follows
$$
\begin{array}{cl}
\T'_{\S}= & \left\{x\in\A^{\Z^d}:\forall n\in\N,\exists
y\in\A^{\Z^d},\exists(\vect{s}_0,\dots,\vect{s}_{n-1})\in\left(\S^{\Z^d}
\right)^n,\right.\\
 & \left.\vect{s}_0\circ\dots\circ \vect{s}_{n-1} (y)=x\right\}.
\end{array}
$$

\begin{remark}
One has both $\T_\Seq\subseteq\T_\S$ and $\T'_\Seq\subseteq\T'_\S$ for any
sequence $\Seq$.
\end{remark}

\section{Realization by sofic subshifts}\label{section.realization}

\subsection{Mozes theorem and the property A}\label{subsection.mozes}

In~\cite{mozes1989tss} Mozes studied non deterministic multidimensional
substitutions, and proved that provided a non deterministic substitution
satisfies a good property, then the subshift it generates is sofic.

All substitutions we consider here are deterministic since the substitutions
rules are given by a function. Nevertheless this formalism provides a way to
study non deterministic substitutions. Given $s$ a non deterministic
substitution, if a letter $a\in\A$ has two patterns $p_1,p_2$ as images, one
replaces $s$ by $s_1$ and $s_2$, where $s_1$ has the same substitutions rules as
$s$ without the rule $a\rightarrow p_2$, and $s_2$ has the same substitutions
rules as $s$ without the rule $a\rightarrow p_1$. By iterating this process, we
can transform a non deterministic substitution into a set $\S$ of deterministic
substitutions, so that the subshift called $(\Omega,\Z^2)$ by Mozes is exactly
the subshift $\T_\S$.

\begin{theorem}[\cite{mozes1989tss}]\label{theorem:Mozes}
Let $\S$ be a set of non degenerate deterministic multidimensional substitutions
of type $A$. Then the subshift $\T_\S$ is a sofic. 
\end{theorem}

We say that a set of substitutions $\S$ is \define{of type $A$}, or \define{has
the property $A$}, if it satisfies the following condition. Let
$p=\vect{u}_1\circ\dots\circ\vect{u}_k(a)$ be a $\S$-pattern $p$ and $l$ a
$2\times2$ pattern that appears in $p$. Suppose there
exists a sequence of patterns of substitutions $\vect{s}_1,\dots,\vect{s}_n$
compatible with the $2\times2$ pattern $l$ that produce a sequence of patterns
$l_0=l,l_1=\vect{s}_1(l_0),\dots,l_n=\vect{s}_n(l_{n-1})$. Then it is possible
to find a sequence of patterns of substitution $\vect{s}'_1,\dots,\vect{s}'_n$
compatible with the pattern $p$ such that the blocks that derive from $l$ in
$p_0=p,p_1=\vect{s}'_1(p_0),\dots,p_n=\vect{s}'_n(p_{n-1})$ are exactly
$l_0,l_1,\dots,l_n$ (see Figure~\ref{figure.typeA}).

\begin{remark}
This property A for sets of substitutions is actually not very restrictive. For
instance any set of substitutions $\S$ such that for every substitution
$s\in\S$, the support of $s(a)$ is the same for any $a\in\A$, has the property
$A$. Moreover, if the set $\S$ is reduced to a single deterministic substitution
$s$, then $\S$ is of type $A$.
\end{remark}

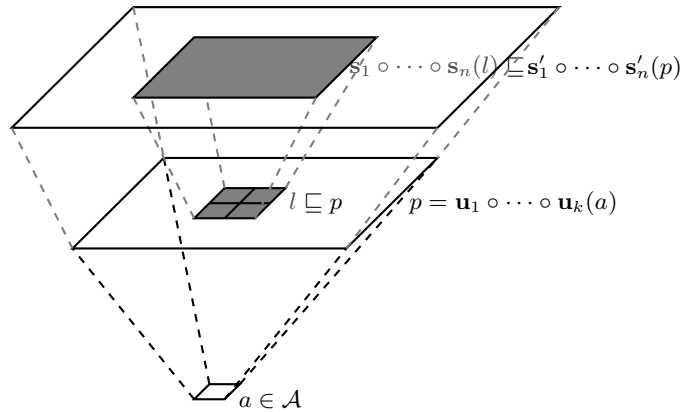
\begin{figure}
\begin{center}
\begin{tikzpicture}[scale=0.4]

\draw[fill=white,thick] (0,0)--(1,0)--(1.5,0.5)--(0.5,0.5)--cycle;
\draw (2.5,0) node{$a\in\A$};
\draw[fill=white,thick] (-4,5)--(5,5)--(8,8)--(-1,8)--cycle;
\draw (10.5,6.5) node{$p=\vect{u}_1\circ\dots\circ\vect{u}_k(a)$};

\draw [thick,dashed] (0,0)--(-4,5);
\draw [thick,dashed] (1,0)--(5,5);
\draw [thick,dashed] (1.5,0.5)--(8,8);
\draw [thick,dashed] (0.5,0.5)--(-1,8);

\draw[fill=black!50,thick] (0,6)--(2,6)--(3,7)--(1,7)--cycle;
\draw (8,11) node{$\textcolor{black!70}{\vect{s}_1\circ\dots\circ\vect{s}_n(l)}\sqsubseteq$};
\draw[thick] (1,6)--(2,7);
\draw[thick] (0.5,6.5)--(2.5,6.5);

\draw[fill=black!50,thick] (-2,10)--(4,10)--(6,12)--(0,12)--cycle;
\draw (4,6.5) node{$\textcolor{black!70}{l}\sqsubseteq p$};
\draw [color=black!50,thick,dashed] (0,6)--(-2,10);
\draw [color=black!50,thick,dashed] (2,6)--(4,10);
\draw [color=black!50,thick,dashed] (3,7)--(6,12);
\draw [color=black!50,thick,dashed] (1,7)--(0,12);

\draw[thick] (-6,9)--(8,9)--(12,13)--(-2,13)--cycle;
\draw (13.5,11) node{$\vect{s}'_1\circ\dots\circ\vect{s}'_n(p)$};
\draw [color=black!50,thick,dashed] (-6,9)--(-4,5);
\draw [color=black!50,thick,dashed] (8,9)--(5,5);
\draw [color=black!50,thick,dashed] (12,13)--(8,8);
\draw [color=black!50,thick,dashed] (-2,13)--(-1,8);

\end{tikzpicture}
\caption{A set of substitutions $\S$ having the property $A$.}
\label{figure.typeA}
\end{center}
\end{figure}

It is possible that the composition of substitutions rules chosen for $l$ is
not compatible with the pattern $p$, and in this case it is possible to find
another sequence of substitution rules compatible with $p$ and such that the
the blocks that derive from $l$ are exactly the $l_0=l,l_1,\dots,l_n$.

Actually one can easily adapt Mozes proof to get a similar result for the
subshift $\T'_\S$. Moreover we do not require the set of substitutions $\S$ to
have the property $A$.

\begin{corollary}\label{corollary.mozes}
Let $\S$ be a set of deterministic multidimensional substitutions. Then the
subshift $\T'_\S$ is sofic. 
\end{corollary}

\textit{Idea of the proof.} We here give some ideas to adapt of the proof of
Theorem~\ref{theorem:Mozes} to obtain a sketch of the proof of
Corollary~\ref{corollary.mozes}. Let $\S$ be a set of substitutions of type
$A$. Mozes constructs a sofic subshift $\Sigma$ such that $\T_\S$ is a factor of
$\Sigma$. The subshift $\Sigma$ contains a grid that ensures that a
configuration $x$ is in the sofic subshift $\Sigma$ if and only if one can find,
for any $n\in\N$, a sequence of infinite patterns of substitutions
$\vect{s}_0,\dots,\vect{s}_{n-1}\in\S^{\Z^d}$ and a
configuration $y_n$ such that $\vect{s}_0\circ\dots\circ \vect{s}_{n-1}
(y_n)=x$. In $\Sigma$ all the $y_n$ are coded in a hierarchical structure. Let
$Q$ be the set of $2\times2$ patterns that appear in a $\S$-pattern. There is an
additional condition : any $2\times2$ pattern that appears in any
configuration $y_n$ is in $Q$. This construction works: given a
configuration $x\in\T_\S$ it is easy to construct a $y\in\Sigma$ that encodes
$x$ and all its pre-images. Reciprocally given a pattern $p$ that appears in a
configuration $x\in\Sigma$, one can find a sequence of finite patterns of
substitutions $(\vect{s}_0,\dots,\vect{s}_{n-1})$ such that $p$ appears in
$\vect{s}_0\circ\dots\circ \vect{s}_{n-1} (p')$, where
$p'$ is either a letter or a $2\times1$, a $1\times2$ or a $2\times2$
pattern. If $p'$ is a letter then $p$ appears in a $\S$-pattern.
Otherwise, $p'$ appears in a $2\times2$ pattern that appears itself
in a $\S$-pattern -- thanks to condition $Q$ --, hence the property $A$ ensures
that $p$ also appears in a $\S$-pattern, that is to say generated by one letter
$a$ (see Figure~\ref{figure.typeA}). So any pattern appearing in $x$ appears in
a $\S$-pattern.

So both type $A$ condition and $Q$ condition are made to ensure that any
pattern that appears in a configuration $x$ also appears in a $\S$-pattern.

The difference between the subshifts $\T_\S$ and $\T'_\S$ is that we remove the
condition that forces a pattern appearing in a configuration $x$ to be a
$\S$-pattern -- and of course we still require that $x$ has a pre-image of any
order by $\S$. Hence the condition of having the property $A$ is no longer
needed, and if we adapt Mozes construction with replacing the set $Q$ by the set
of all the $2\times2$ patterns, then $\T'_\S$ is a factor of the sofic subshift
obtained. This proves the corollary.

%%%%%%%%%%%%%%%%%%%%%%
%
%%%%%%%%%%%%%%%%%%%%%%

\subsection{Effective S-adic subshifts are sofic}

Let $\Seq\in\S^{\N}$, of course one has $\T_{\Seq}\subset\T_{\S}$, but there is
no immediate reason for $\T_{\Seq}$ to be also sofic, and moreover for
cardinality reasons there exists non-sofic S-adic subshifts. 

\begin{theorem}\label{theorem.Sadic}
Let $\S$ be a finite set of non degenerate multidimensional substitutions and
$\Seq\in\S^{\N}$ be an effective sequence. Then $\T'_\Seq$ is sofic. If $\S$ has
the property $A$, then $\T_\Seq$ is sofic. 
\end{theorem}

\remark{We only present the proof that $\T_\Seq$ is sofic. The proof is similar
for $\T'_\Seq$, one just need to replace $\T_{\S}$ by $\T'_{\S}$.}

\begin{proof}
We now assume that $d=2$, the proof is similar for $d\geq3$. Let $\S$ be a
finite set of non degenerate $(\A,2)$-substitutions, we define
$\A'=\A\times\S^2$. To every $s\in\S$ we associate a $(\A',2)$-substitution
$\widetilde{s}$ with same support such that 
\scriptsize
$$
\hspace{-1.5cm}
\widetilde{s}:(a,s_V,s_H)\mapsto
\begin{array}{c|ccc|c}
(s(a)_{(0,\vect{k}_2^s(a))},s,s_H) & s(a)_{(1,\vect{k}_2^s(a))},s,s_H) & \dots &
s(a)_{(\vect{k}_1^s(a)-1,\vect{k}_2^s(a))},s,s_H) &
(s(a)_{(\vect{k}_1^s(a),\vect{k}_2^s(a))},s_V,s_H)\\
\hline
(s(a)_{(0,\vect{k}_2^s(a)-1)},s,s) & & & & (s(a)_{(\vect{k}_1^s(a),\vect{k}_2^s(a)-1)},s_V,s) \\
\vdots & & (s(a)_{(i,j)},s,s) & & \vdots \\
(s(a)_{(0,1)},s,s) & & & & (s(a)_{(\vect{k}_1^s(a),1)},s_V,s) \\
\hline
(s(a)_{(0,0)},s,s) & (s(a)_{(1,0)},s,s) & \dots & (s(a)_{(\vect{k}_1^s(a)-1,0)},s,s) &
(s(a)_{(\vect{k}_1^s(a),0)},s_V,s)
\end{array}
$$
\normalsize

All these substitutions $\widetilde{s}$ form a set
$\widetilde{\S}=\left\{\widetilde{s}: s\in\S\right\}$. Let
$\Seq=(s_i)_{i\in\N}\in\S^{\N}$ be an effective sequence, we can thus consider
the effective sequence
$\widetilde{\Seq}=(\widetilde{s}_i)_{i\in\N}\in\widetilde{\S}^{\N}$. The aim of
substitutions $\widetilde{s}$ is to keep a record of the sequence of
substitutions previously applied.

\begin{example}\label{example.4}
Let $\S$ be the set of 2-dimensional substitutions on the alphabet
$\A=\{\circ,\bullet\}$ defined in Example~\ref{example.3}. On the following
picture where
$\widetilde{\Seq}=(\widetilde{s}_d,\widetilde{s}_d,\widetilde{s}_a,\widetilde{s}
_a,\dots)$ applied on the letter $\bullet$, one can find on the bottom line of
the patterns $s_2(\bullet),s_1\circ s_2(\bullet),s_0\circ s_1\circ
s_2(\bullet),\dots$ the sequence of the substitutions already applied appears.

$$
\hspace{-1cm}
(\bullet,s_3,s_3) \overset{\textcolor{blue}{\widetilde{s}_2}}{\longmapsto}
\scriptsize
\begin{array}{cc}
(\circ,\textcolor{blue}{s_2},s_3) & (\circ,s_3,s_3) \\
(\bullet,\textcolor{blue}{s_2},\textcolor{blue}{s_2}) &
(\circ,s_3,\textcolor{blue}{s_2})
\end{array}
\normalsize
\overset{\textcolor{red}{\widetilde{s}_1}}{\longmapsto}
\scriptsize
\begin{array}{cc|cc}
(\circ,\textcolor{red}{s_1},s_3) & (\circ,\textcolor{blue}{s_2},s_3) &
(\circ,\textcolor{red}{s_1},s_3) & (\circ,s_3,s_3) \\
(\circ,\textcolor{red}{s_1},\textcolor{red}{s_1}) &
(\circ,\textcolor{blue}{s_2},\textcolor{red}{s_1}) &
(\circ,\textcolor{red}{s_1},\textcolor{red}{s_1}) &
(\circ,s_3,\textcolor{red}{s_1}) \\
\hline
(\circ,\textcolor{red}{s_1},\textcolor{blue}{s_2}) &
(\circ,\textcolor{blue}{s_2},\textcolor{blue}{s_2}) &
(\circ,\textcolor{red}{s_1},\textcolor{blue}{s_2}) &
(\circ,s_3,\textcolor{blue}{s_2}) \\
(\bullet,\textcolor{red}{s_1},\textcolor{red}{s_1}) &
(\circ,\textcolor{blue}{s_2},\textcolor{red}{s_1}) &
(\circ,\textcolor{red}{s_1},\textcolor{red}{s_1}) &
(\circ,s_3,\textcolor{red}{s_1})
\end{array}
\normalsize
$$
$$
\hspace{-1.5cm}
\overset{\textcolor{violet}{\widetilde{s}_0}}{\longmapsto}
\scriptsize
\begin{array}{ccc|ccc|ccc|ccc}
(\bullet,\textcolor{violet}{s_0},s_3) &
(\bullet,\textcolor{violet}{s_0},s_3) & (\bullet,\textcolor{red}{s_1},s_3) &
(\bullet,\textcolor{violet}{s_0},s_3) &
(\bullet,\textcolor{violet}{s_0},s_3) & (\bullet,\textcolor{blue}{s_2},s_3) &
(\bullet,\textcolor{violet}{s_0},s_3) &
(\bullet,\textcolor{violet}{s_0},s_3) &
(\bullet,\textcolor{red}{s_1},s_3) & (\bullet,\textcolor{violet}{s_0},s_3) &
(\bullet,\textcolor{violet}{s_0},s_3) & (\bullet,s_3,s_3) \\
(\bullet,\textcolor{violet}{s_0},\textcolor{violet}{s_0}
) &
(\bullet,\textcolor{violet}{s_0},\textcolor{violet}{s_0}
) & (\bullet,\textcolor{red}{s_1},\textcolor{violet}{s_0}) &
(\bullet,\textcolor{violet}{s_0},\textcolor{violet}{s_0}
) &
(\bullet,\textcolor{violet}{s_0},\textcolor{violet}{s_0}
) & (\bullet,\textcolor{blue}{s_2},\textcolor{violet}{s_0}) &
(\bullet,\textcolor{violet}{s_0},\textcolor{violet}{s_0}
) &
(\bullet,\textcolor{violet}{s_0},\textcolor{violet}{s_0}
) &
(\bullet,\textcolor{red}{s_1},\textcolor{violet}{s_0}) &
(\bullet,\textcolor{violet}{s_0},\textcolor{violet}{s_0}
) &
(\bullet,\textcolor{violet}{s_0},\textcolor{violet}{s_0}
) & (\bullet,s_3,\textcolor{violet}{s_0}) \\
(\circ,\textcolor{violet}{s_0},\textcolor{violet}{s_0})
&
(\circ,\textcolor{violet}{s_0},\textcolor{violet}{s_0})
& (\circ,\textcolor{red}{s_1},\textcolor{violet}{s_0}) &
(\circ,\textcolor{violet}{s_0},\textcolor{violet}{s_0})
&
(\circ,\textcolor{violet}{s_0},\textcolor{violet}{s_0})
& (\circ,\textcolor{blue}{s_2},\textcolor{violet}{s_0}) &
(\circ,\textcolor{violet}{s_0},\textcolor{violet}{s_0})
&
(\circ,\textcolor{violet}{s_0},\textcolor{violet}{s_0})
& (\circ,\textcolor{red}{s_1},\textcolor{violet}{s_0}) &
(\circ,\textcolor{violet}{s_0},\textcolor{violet}{s_0})
&
(\circ,\textcolor{violet}{s_0},\textcolor{violet}{s_0})
& (\circ,s_3,\textcolor{violet}{s_0}) \\
\hline
(\bullet,\textcolor{violet}{s_0},\textcolor{red}{s_1}) &
(\bullet,\textcolor{violet}{s_0},\textcolor{red}{s_1}) &
(\bullet,\textcolor{red}{s_1},\textcolor{red}{s_1}) &
(\bullet,\textcolor{violet}{s_0},\textcolor{red}{s_1}) &
(\bullet,\textcolor{violet}{s_0},\textcolor{red}{s_1}) &
(\bullet,\textcolor{blue}{s_2},\textcolor{red}{s_1}) &
(\bullet,\textcolor{violet}{s_0},\textcolor{red}{s_1}) &
(\bullet,\textcolor{violet}{s_0},\textcolor{red}{s_1}) &
(\bullet,\textcolor{red}{s_1},\textcolor{red}{s_1}) &
(\bullet,\textcolor{violet}{s_0},\textcolor{red}{s_1}) &
(\bullet,\textcolor{violet}{s_0},\textcolor{red}{s_1}) &
(\bullet,s_3,\textcolor{red}{s_1}) \\
(\bullet,\textcolor{violet}{s_0},\textcolor{violet}{s_0}
) &
(\bullet,\textcolor{violet}{s_0},\textcolor{violet}{s_0}
) & (\bullet,\textcolor{red}{s_1},\textcolor{violet}{s_0}) &
(\bullet,\textcolor{violet}{s_0},\textcolor{violet}{s_0}
) &
(\bullet,\textcolor{violet}{s_0},\textcolor{violet}{s_0}
) & (\bullet,\textcolor{blue}{s_2},\textcolor{violet}{s_0}) &
(\bullet,\textcolor{violet}{s_0},\textcolor{violet}{s_0}
) &
(\bullet,\textcolor{violet}{s_0},\textcolor{violet}{s_0}
) &
(\bullet,\textcolor{red}{s_1},\textcolor{violet}{s_0}) &
(\bullet,\textcolor{violet}{s_0},\textcolor{violet}{s_0}
) &
(\bullet,\textcolor{violet}{s_0},\textcolor{violet}{s_0}
) & (\bullet,s_3,\textcolor{violet}{s_0}) \\
(\circ,\textcolor{violet}{s_0},\textcolor{violet}{s_0})
&
(\circ,\textcolor{violet}{s_0},\textcolor{violet}{s_0})
& (\circ,\textcolor{red}{s_1},\textcolor{violet}{s_0}) &
(\circ,\textcolor{violet}{s_0},\textcolor{violet}{s_0})
&
(\circ,\textcolor{violet}{s_0},\textcolor{violet}{s_0})
& (\circ,\textcolor{blue}{s_2},\textcolor{violet}{s_0})
&
(\circ,\textcolor{violet}{s_0},\textcolor{violet}{s_0})
&
(\circ,\textcolor{violet}{s_0},\textcolor{violet}{s_0})
& (\circ,\textcolor{red}{s_1},\textcolor{violet}{s_0}) &
(\circ,\textcolor{violet}{s_0},\textcolor{violet}{s_0})
&
(\circ,\textcolor{violet}{s_0},\textcolor{violet}{s_0})
&
(\circ,s_3,\textcolor{violet}{s_0}) \\
\hline
(\bullet,\textcolor{violet}{s_0},\textcolor{blue}{s_2}) &
(\bullet,\textcolor{violet}{s_0},\textcolor{blue}{s_2}) &
(\bullet,\textcolor{red}{s_1},\textcolor{blue}{s_2}) &
(\bullet,\textcolor{violet}{s_0},\textcolor{blue}{s_2}) &
(\bullet,\textcolor{violet}{s_0},\textcolor{blue}{s_2}) &
(\bullet,\textcolor{blue}{s_2},\textcolor{blue}{s_2}) &
(\bullet,\textcolor{violet}{s_0},\textcolor{blue}{s_2}) &
(\bullet,\textcolor{violet}{s_0},\textcolor{blue}{s_2}) &
(\bullet,\textcolor{red}{s_1},\textcolor{blue}{s_2}) &
(\bullet,\textcolor{violet}{s_0},\textcolor{blue}{s_2}) &
(\bullet,\textcolor{violet}{s_0},\textcolor{blue}{s_2}) &
(\bullet,s_3,\textcolor{blue}{s_2}) \\
(\bullet,\textcolor{violet}{s_0},\textcolor{violet}{s_0}
) &
(\bullet,\textcolor{violet}{s_0},\textcolor{violet}{s_0}
) & (\bullet,\textcolor{red}{s_1},\textcolor{violet}{s_0}) &
(\bullet,\textcolor{violet}{s_0},\textcolor{violet}{s_0}
) &
(\bullet,\textcolor{violet}{s_0},\textcolor{violet}{s_0}
) & (\bullet,\textcolor{blue}{s_2},\textcolor{violet}{s_0}) &
(\bullet,\textcolor{violet}{s_0},\textcolor{violet}{s_0}
) &
(\bullet,\textcolor{violet}{s_0},\textcolor{violet}{s_0}
) &
(\bullet,\textcolor{red}{s_1},\textcolor{violet}{s_0}) &
(\bullet,\textcolor{violet}{s_0},\textcolor{violet}{s_0}
) &
(\bullet,\textcolor{violet}{s_0},\textcolor{violet}{s_0}
) & (\bullet,s_3,\textcolor{violet}{s_0}) \\
(\circ,\textcolor{violet}{s_0},\textcolor{violet}{s_0})
&
(\circ,\textcolor{violet}{s_0},\textcolor{violet}{s_0})
& (\circ,\textcolor{red}{s_1},\textcolor{violet}{s_0}) &
(\circ,\textcolor{violet}{s_0},\textcolor{violet}{s_0})
&
(\circ,\textcolor{violet}{s_0},\textcolor{violet}{s_0})
& (\circ,\textcolor{blue}{s_2},\textcolor{violet}{s_0}) &
(\circ,\textcolor{violet}{s_0},\textcolor{violet}{s_0})
&
(\circ,\textcolor{violet}{s_0},\textcolor{violet}{s_0})
& (\circ,\textcolor{red}{s_1},\textcolor{violet}{s_0}) &
(\circ,\textcolor{violet}{s_0},\textcolor{violet}{s_0})
&
(\circ,\textcolor{violet}{s_0},\textcolor{violet}{s_0})
& (\circ,s_3,\textcolor{violet}{s_0}) \\
\hline
(\circ,\textcolor{violet}{s_0},\textcolor{red}{s_1}) &
(\circ,\textcolor{violet}{s_0},\textcolor{red}{s_1}) &
(\circ,\textcolor{red}{s_1},\textcolor{red}{s_1}) &
(\bullet,\textcolor{violet}{s_0},\textcolor{red}{s_1}) &
(\bullet,\textcolor{violet}{s_0},\textcolor{red}{s_1}) &
(\bullet,\textcolor{blue}{s_2},\textcolor{red}{s_1}) &
(\bullet,\textcolor{violet}{s_0},\textcolor{red}{s_1}) &
(\bullet,\textcolor{violet}{s_0},\textcolor{red}{s_1}) &
(\bullet,\textcolor{red}{s_1},\textcolor{red}{s_1}) &
(\bullet,\textcolor{violet}{s_0},\textcolor{red}{s_1}) &
(\bullet,\textcolor{violet}{s_0},\textcolor{red}{s_1}) &
(\bullet,s_3,\textcolor{red}{s_1}) \\
(\circ,\textcolor{violet}{s_0},\textcolor{violet}{s_0})
&
(\circ,\textcolor{violet}{s_0},\textcolor{violet}{s_0})
& (\circ,\textcolor{red}{s_1},\textcolor{violet}{s_0}) &
(\bullet,\textcolor{violet}{s_0},\textcolor{violet}{s_0}
) &
(\bullet,\textcolor{violet}{s_0},\textcolor{violet}{s_0}
) & (\bullet,\textcolor{blue}{s_2},\textcolor{violet}{s_0}) &
(\bullet,\textcolor{violet}{s_0},\textcolor{violet}{s_0}
) &
(\bullet,\textcolor{violet}{s_0},\textcolor{violet}{s_0}
) &
(\bullet,\textcolor{red}{s_1},\textcolor{violet}{s_0}) &
(\bullet,\textcolor{violet}{s_0},\textcolor{violet}{s_0}
) &
(\bullet,\textcolor{violet}{s_0},\textcolor{violet}{s_0}
) & (\bullet,s_3,\textcolor{violet}{s_0}) \\
(\bullet,\textcolor{violet}{s_0},\textcolor{violet}{s_0}
) &
(\bullet,\textcolor{violet}{s_0},\textcolor{violet}{s_0}
) & (\bullet,\textcolor{red}{s_1},\textcolor{violet}{s_0}) &
(\circ,\textcolor{violet}{s_0},\textcolor{violet}{s_0})
&
(\circ,\textcolor{violet}{s_0},\textcolor{violet}{s_0})
& (\circ,\textcolor{blue}{s_2},\textcolor{violet}{s_0}) &
(\circ,\textcolor{violet}{s_0},\textcolor{violet}{s_0})
&
(\circ,\textcolor{violet}{s_0},\textcolor{violet}{s_0})
& (\circ,\textcolor{red}{s_1},\textcolor{violet}{s_0}) &
(\circ,\textcolor{violet}{s_0},\textcolor{violet}{s_0})
&
(\circ,\textcolor{violet}{s_0},\textcolor{violet}{s_0})
& (\circ,s_3,\textcolor{violet}{s_0}) 
\end{array}
$$
\end{example}

One considers $\pi:\A'\to\A$ the one-block map which keeps the letter of $\A$
and $\pi_V:\A'\to\S$ (resp. $\pi_H:\A'\to\S$) the one-block map which keeps the
substitution $s_V\in\S$ (resp. $s_H\in\S$) of an element $(a,s_V,s_H)\in\A'$.

\Claim{$\T_\Seq=\pi\left(\T_{\widetilde{\Seq}}\right)$}
\bclaimprf
This is straightforward, since the alphabet $\A'$ contains alphabet $\A$, and
substitution $\widetilde{s}$ restricted to alphabet $\A$ is exactly substitution
$s$.
\eclaimprf

Consequently, it is sufficient to prove that $\T_{\widetilde{\Seq}}$ is sofic.

\Claim{The subshift $\Sigma=\sa{(1,0\dots,0)\Z}{\pi_V(\T_{\widetilde{\Seq}})}$
is effective.}
\bclaimprf
The class of effective subshifts is closed under factor, but also under
projective subaction. This follows from the fact that projective subactions are
special cases of factors of subactions, and by Theorem~3.1 and Proposition~3.3
of ~\cite{hochman2007drp} which establish that symbolic factors and subactions
preserve effectiveness. That is to say $\mathcal{C}l_{\SA}(\reshift)=\reshift$.
Thus it is sufficient to prove that $\T_{\widetilde{\Seq}}$ is an effective
subshift. The sequence of substitutions $\Seq$ is effective so it is the same
for the sequence $\widetilde{\Seq}$. Consequently, one can design an algorithm
that computes the $\widetilde{\Seq}$-patterns, which proves that the subshift
$\T_{\widetilde{\Seq}}$ is effective.
\eclaimprf

By Theorem~\ref{theorem.simulation} there exists a
$d$-dimensional subshift of finite type $\T_{\Sigma}$ on an alphabet $\B$ and a
factor $\pi_{\Sigma}:\B^{\Z^d}\to\S^{\Z^d}$ such that 
$\Sigma=\pi_{\Sigma}\left(\sa{(1,0,\dots,0)\Z}{\T_{\Sigma}}\right)$. Note that
the fact that $d\geq2$ is crucial here, since the previous statement is not true
for $d=1$.

If we consider a configuration of the subshift $\T_{\widetilde{\S}}$ defined in
Section~\ref{subsection.nondeterm}, any substitution that appears may be chosen
in the set $\S$, provided it is still compatible with the configuration. But on
a same level, different substitutions may appear, which does not fit with the
S-adic subshift definition. To solve this problem we synchronize substitutions
so that the same substitution is used on a given level. To do that we need to
ensure for any configuration $x\in\T_{\widetilde{\S}}$, the same substitution
appears in $\pi_V(x)$ on each row (resp. in $\pi_H(x)$ on each column). 

We thus define the subshift $\widetilde{\T}_{\widetilde{\S}}$ on the following
way :  
$$\widetilde{\T}_{\widetilde{\S}}=\left\{
x\in\T_{\widetilde{\S}}:\forall(i,j)\in\Z^2,
\pi_H(x)_{(i,j)}=\pi_H(x)_{(i,j+1)}\text{ and
}\pi_V(x)_{(i,j)}=\pi_V(x)_{(i+1,j)} \right\}.$$

We deduce that 
$$\widetilde{\T}_{\widetilde{\S}}=\bigcup_{\widetilde{\Seq}\in\widetilde{\S}^\N}
\T_{\widetilde{\Seq}}\subset\T_{\widetilde{\S}}$$

Finally we consider
$$\T_{\texttt{Final}}=\left\{(x,s) \in
\widetilde{\T}_{\widetilde{\S}}\times\T_{\Sigma}: \forall (i,j)\in\Z^2,\,
\pi_{V}(x)_{(i,j)}=\pi_{\Sigma}(s)_{(i,j)} \right\}.$$

Thanks to Corollary~\ref{corollary.mozes}, we know that the subshift
$\T_{\widetilde{\S}}$ is sofic, hence so is $\widetilde{\T}_{\widetilde{\S}}$.
Hence by construction, $\T_{\texttt{Final}}$ is
a sofic subshift. Consider the letter-to-letter factor map
$\pi_{\texttt{Final}}:\T_{\texttt{Final}}\to\A'^{\Z^d}$ which
keeps the letter of $\A'$.

\Claim{$\pi_{\texttt{Final}}\left(\T_{\texttt{Final}}\right)=\T_{\widetilde{\Seq
}}$}
\bclaimprf 
Given a configuration $x\in\T_{\widetilde{\Seq}}$, it is easy to construct a
corresponding element in $\T_{\texttt{Final}}$. Reciprocally, suppose you are
given a configuration $x_\texttt{Final}\in\T_{\texttt{Final}}$. Replacing
substitutions in $\S$ by composition of two substitutions of $\S$ if necessary,
we assume that for all $s\in\S$ and all $a\in\A$,
$\vect{k}^s_1(a),\vect{k}^s_2(a)\geq2$. First the
$\widetilde{\T}_{\widetilde{\S}}$ part of $x_\texttt{Final}$ ensures that
$\pi_\texttt{Final}(x_\texttt{Final})$ is an element of one
$\T_{\widetilde{\Seq'}}$ for some $\widetilde{\Seq}'\in\widetilde{S}^\N$.
Secondly the condition that links the $\widetilde{\T}_{\widetilde{\S}}$ part
with the $\T_{\Sigma}$ part certifies that $\Seq'=\Seq$: substitution $s_0$ is
the only one which is repeated at least twice systematically -- since
$\vect{k}^s_1,\vect{k}^s_2\geq2$. If we apply the same reasoning to a pre-image
of $x_\texttt{Final}$ by $s_0$, we can find $s_1$ and so on.
\eclaimprf
\end{proof}

\section{Effective subshifts which are S-adic}

We want to find a reciprocal statement to Theorem~\ref{theorem.Sadic}: what can
we say about an effective subshift which is S-adic? A set of substitutions $\S$
has \emph{unique derivation} if for every element 
$x\in\widetilde{\T}_\S=\bigcup_{\Seq\in\S^\N}\T_\Seq$, there exist a unique
$s\in\S$, a unique $y\in\A^{\Z^d}$ and a unique
$i\in\bigcup_{a\in\A}\U_{\vect{k}^s(a)}$ such that
$s_\infty(y)=\sigma^i(x)$.

\begin{theorem}
Let $\S$ be a set of substitutions with unique derivation and let
$\Seq\in\S^\N$. If the S-adic subshift $\T_\Seq$ is effective (and in particular
if $\T_\Seq$ is sofic) then $\Seq$ is effective.
\end{theorem}

\begin{proof}
 Since $\T_\Seq$ is effective, there exists a Turing machine $\mathcal{M}$ that
enumerates all its forbidden patterns. We compute the sets
$\mathcal{E}_s=\left\{ s(a):a\in\A\right\}$, and we try to find which
substitution $s\in\S$ is the first of the sequence $\Seq$. To do that, for every
$s\in\S$, we try to partition $\Z^d$ with patterns from one $\mathcal{E}_s$,
so that no pattern enumerated by the machine $\mathcal{M}$ appear. These
partitions are made in parallel, and the unique derivation condition ensures
that only one of them will work. The calculation stops when all the
substitutions but one has been rejected -- this always happens -- and this
substitution is $s_0$. Apply again this process to a pre-image of $x$ by
$s$, to get the next substitution $s_1$, and so on. 
\end{proof}

\bibliographystyle{alpha}
\bibliography{biblio}

\end{document}